\documentclass[conference]{IEEEtran}
\usepackage{cite}
\usepackage{amsthm}
\usepackage{amsmath,amssymb,amsfonts}
\usepackage{amsthm}
\usepackage{algorithmic}
\usepackage{graphicx}
\usepackage{textcomp}
\usepackage{xcolor}
\usepackage{caption}
\newtheorem{theorem}{Theorem}
\newtheorem{res}{Result}
\newtheorem{cor}{Corollary}
\usepackage{comment}
\usepackage{bm}  

\belowdisplayskip=\abovedisplayskip
\begin{document}

\title{Performance Analysis of Movable Antenna Arrays}


\author{\IEEEauthorblockN{Gayani Siriwardana$^*$, Peter J. Smith$^\dagger$, Himal A. Suraweera$^\ddag$,  Rajitha~Senanayake$^*$}
\IEEEauthorblockA{$^*$Department of Electrical
and Electronic Engineering, University of Melbourne, Parkville,
VIC. 3010, Australia.}
\IEEEauthorblockA{$^\dagger$School of Mathematics and Statistics, Victoria University of Wellington, Wellington, New Zealand.}
\IEEEauthorblockA{$^\ddag$Department of Electrical and Electronic Engineering, University of Peradeniya, Sri Lanka.}
\small{gsiriwardana@student.unimelb.edu.au,
peter.smith@vuw.ac.nz,
himal@eng.pdn.ac.lk,
rajitha.senanayake@unimelb.edu.au}}

\maketitle

\begin{abstract}
This paper provides a thorough mathematical analysis of continuous movable antenna (MA) arrays. Focusing on the multiple antenna case, we consider a linear antenna array with multiple fixed antenna elements that moves along  a line. We assume a full, spatially coherent correlation model and continuous positioning of the array. We provide asymptotically exact approximations to the upper tail of the cumulative distribution function (cdf) of the signal-to-noise ratio (SNR), considering both correlated and uncorrelated antenna elements in the array. We also obtain a novel closed-form expression for the level crossing rate (LCR) of the SNR under correlated array elements, where a non-separable two-dimensional correlation is present. The analysis is validated through simulations, confirming both the accuracy of the LCR expressions and the tightness of the cdf bounds in the upper tail. Numerical results show that the proposed MA array outperforms single fluid antenna and fixed-array systems, with reduced inter-element spacing providing further performance gains.
\end{abstract}
\section{Introduction}

Innovative reconfigurable antenna technologies, such as fluid antenna systems (FASs)~\cite{WK_New_Surv_Turor_25}, have been proposed to maximize the benefits of the spatial domain in wireless communication. A fluid antenna (FA) refers to a software-controlled structure composed of fluidic, conductive, or dielectric materials, designed to dynamically reshape and reposition itself. This broad concept includes various types of position-flexible antennas, such as on-off switching pixels, liquid-based structures, and movable antenna (MA) systems~\cite{WK_New_Surv_Turor_25}. 

In recent years, the performance of FAs has been studied for various systems. In~\cite{Wong_ITWC_2021}, the performance of a single FAS with discrete positioning is analyzed, considering a fixed set of positions. In~\cite{Psomas_ICL_2023}, single FAS with continuous positioning is analyzed, where the antenna can move freely along the available space. The use of single FAS has also been extended to multi-user scenarios by exploiting the spatial variations of fading channels to mitigate inter-user interference~\cite{wong_TWC_2022}.

To further improve system performance, recent works have extended to multi-antenna FAS configurations. Existing work on multi-antenna FASs can be broadly classified as discrete and continuous positioning systems. In discrete positioning~\cite{XLai_COMML_2024,Wee_Kiat_TWC_2024,Krikidis_2024Dec_TC}, multiple ports can be activated simultaneously, while in continuous positioning~\cite{Ma_TWC2024,Xiao_2024_TWC,Zhu_2025sept_TWC}, multiple antennas can move along the available space, which is more challenging to analyze. In~\cite{XLai_COMML_2024}, the outage probability of a FAS with multiple activated ports and maximum ratio combining (MRC) receiver processing has been derived. In~\cite{Wee_Kiat_TWC_2024}, multiple-input multiple-output (MIMO) FAS has been studied, in which both transmitter and receiver are equipped with multi-port FAs. For continuous positioning systems, a MIMO system with multiple MAs at both the transmitter and receiver is proposed in~\cite{Ma_TWC2024}, showing that channel capacity can be improved by up to 30.3\% compared to traditional MIMO systems. Most recently,~\cite{Xiao_2024_TWC,Zhu_2025sept_TWC} consider multiple MAs and multiple movable arrays respectively. Both studies propose algorithms to determine the optimal antenna positions and show that significant gains are achievable through multi-antenna mobility.

Among various multi-antenna configurations in FASs, continuous positioning offers significant potential to exploit maximum spatial diversity. Interestingly, such positioning also provides insights about the limiting performance of FASs. However, most existing works,~\cite{Ma_TWC2024,Xiao_2024_TWC,Zhu_2025sept_TWC}, focus on optimizing antenna positions, and to the best of our knowledge, a thorough mathematical analysis on the system performance is still lacking. Motivated by this gap, we aim to make analytical progress on the system performance of continuous movable antenna arrays. More specifically, we focus on the upper tail of the cumulative distribution function (cdf) of the received signal-to-noise ratio (SNR) and derive analytical expressions to bound the performance. While the lower tail of the SNR cdf is important to assess reliability, here we focus on the upper tail as it quantifies the ability of a FAS to achieve high data rates through maximizing the instantaneous SNR. In particular, for FASs, the upper tail is especially relevant since the antenna searches for a position that maximizes the instantaneous value of the performance metric. The main contributions of this paper can be summarized as follows:
\begin{enumerate}
    \item We analyze a multi-antenna FAS, consisting of a linear MA array whose elements move together continuously along the axis perpendicular to its orientation.
    \item We derive asymptotically exact approximations for the upper tail of the cdf of the SNR, considering both correlated and uncorrelated array elements. A novel expression for the level crossing rate (LCR) of the SNR is also obtained for the scenario of correlated elements.
    \item Based on our results, we draw key insights into the performance of the proposed MA array. In particular, we show that placing the array elements closer together improves the overall system performance.
\end{enumerate}

\section{System Model}
We consider a system with a single-antenna transmitter and a receiver equipped with a linear MA array. The array elements are fixed relative to each other, while the entire array can move along a perpendicular axis, as shown in Fig.~\ref{im1}.
\begin{figure}[tb]
    \centerline{\includegraphics{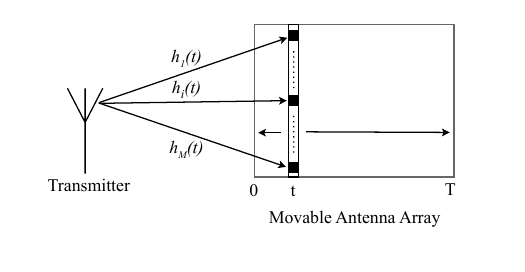}}
    \vspace{-7mm}  
    \caption{An illustration of a movable antenna array.}
    \label{im1}
\end{figure}

The array consists of $M$ fixed antennas with $\Delta$ separation between two adjacent antennas. We assume that the array can occupy any position $t$, from zero to $T$. Furthermore, we ignore delays in antenna movement \cite{Wong_ITWC_2021, wong_TWC_2022} in order to investigate optimal performance. Based on this setup, the received signal at position $t$ at antenna $i$ can be modeled as
\begin{equation}
\label{received_signal}
  y_{i}(t) = h_{i}(t)x + n(t),
\end{equation}
for $0 \leq t \leq T$ where $h_{i}(t)$ denotes the channel between the transmitter and the $i$-th receiver when the array is at position $t$. Since all antennas in the array are closely located, we assume equal channel powers given by $\beta$ and $h_{i}(t) \sim \mathcal{CN}(0,\beta)$ for all $i \in \{1, \ldots, M\}$. The channel vector at location $t \in [0,T]$ is $\mathbf{h}(t) = \left[ h_1(t), \, h_2(t), \, \cdots, \, h_M(t) \right]^T
$ where \((\cdot)^T\) denotes the transpose. The transmitted data symbol is $x$ where $\mathbb{E}\left[|x|^2\right] = E_{x}$ and $n(t)$ is the additive white Gaussian noise (AWGN) at position $t$ with zero mean and variance $\sigma^2$. 

To characterize the correlation between antennas at different positions, we first define the distance between the \(r\)-th antenna of the array at position $t$ and the \(s\)-th antenna of the array at position $t+\tau$ as
\begin{align}\label{dist}
    d_{r,s}(\tau) = \sqrt{\tau^2 + \Delta^2(r-s)^2}.
\end{align}
Note that $\tau$, $\Delta$, and all other distances used in this paper are measured in wavelengths. Assuming the classic isotropic scattering model, the correlation between the channels at two positions separated by a distance $d_{r,s}(\tau)$, can be described using Jakes' model as \cite{Heath_book_2018}
\begin{equation}\label{correlation}
\rho(d_{r,s}(\tau)) = \frac{1}{\beta}\mathbb{E}[h_r(t)h_s^*(t+\tau)] =J_0(2\pi {d_{r,s}}(\tau)),
\end{equation}
where $h_i^*(.)$ denotes the complex conjugate of $h_i(.)$,
and $J_0(\cdot)$ is the zero-th order Bessel function of the first kind \cite[Eq. (8.402)]{Gradshteyn_book_2007}. The correlation between antennas in the antenna array at a fixed position corresponds to the case $\tau=0$, and we define the corresponding spatial correlation matrix as
\begin{align}\label{sig}
    \mathbf{\Sigma} = \frac{1}{\beta}\mathbb{E}[\bm{\mathrm{h}}(t) \bm{\mathrm{h}}(t)^H],
\end{align}
where \((\cdot)^H\) denotes the Hermitian. In addition, by using \eqref{correlation}, the entries of $\mathbf{\Sigma}$ can be written as
\begin{equation}\label{sigmars}
[\mathbf{\Sigma}]_{r,s} =
\begin{cases}
1, & \text{if } r = s,\\[2mm]
 J_0\!\bigl(2\pi |r-s| \Delta\bigr), & \text{if } r \neq s.
\end{cases}
\end{equation}
It should be noted that any correlation model of the form
\begin{align}\label{correlation2}
\rho(\tau) = 1 - b \tau^2 + o(\tau^2),	
\end{align}
for small $\tau$ and $b > 0$ gives a mean-square differentiable channel which corresponds to a physically reasonable Rayleigh channel. Note that \eqref{correlation2} uses the standard Bachmann-Landau little-O notation. Here, we consider the Jakes' model in \eqref{correlation}, for which $b = \pi^2$ using the series representation in \cite{Gradshteyn_book_2007}. Although numerical results are obtained using this model, the analysis in Section~\ref{3} applied to any correlation model satisfying \eqref{correlation2} with a known first derivative of $\rho$. 

The MA array employs MRC as the receiver combiner, and the resulting post-combining SNR can be written as
 \begin{equation}\label{snr_array}
S(t) = \frac{E_{x}}{\sigma^2} \mathbf{h}^H(t) \mathbf{h}(t) = \frac{E_{x}}{\sigma^2}\sum_{i=1}^{M} |h_i(t)|^2.
\end{equation}
 When operating, the antenna array is moved to a position such that the overall SNR, $S(t)$ is maximized. Let us denote the optimal position of the array as $t^*$ such that 
\begin{equation}
S(t^*) = \sup_{0 \leq t \leq T} \{S(t)\},
\end{equation}
where $\sup_{0 \leq t \leq T} \{S(t)\}$ denotes the supremum of $S(t)$. For ease of notation, we denote $S(t^*) \triangleq S^*$ in the following.
\section{Performance analysis}\label{3}

In this section, we analyze the cdf of $S^*$ for the MA array. Deriving the exact distribution of $S^*$ is mathematically challenging due to the continuous nature of the MAs and the full physical modeling of the correlation. Thus, in this paper, following \cite{Psomas_ICL_2023}, we take a more tractable approach and analyze the upper tail of the cdf of $S^*$, denoted by $F_{S^*}(s_{\text{th}})$. The methodology in \cite{Psomas_ICL_2023} gives a lower bound on the cdf which is an asymptotically exact approximation to $F_{S^*}(s_{\text{th}})$ for large arguments, $s_{\text{th}}$. The relevant result is \cite[Eq. (41)]{Psomas_ICL_2023}
\begin{align}\label{s*}
F_{S^*}(s_{\text{th}})  \geq \mathbb{P}(S(0) \leq s_{\text{th}}) - T\times{LCR_{S(t)}(s_{\text{th}})},
\end{align}
where $\mathbb{P}(E)$ is the probability of event $E$ and ${LCR_{S(t)}(s_{\text{th}})}$ is the LCR of the process $S(t)$ across the threshold, $s_{\text{th}}$. This is an asymptotically exact lower bound for any stationary process $S(t)$ satisfying \eqref{correlation2}. Note that \eqref{s*} only depends on the cdf of $S(0)$ and the LCR of $S(t)$. \\

\subsection{Movable Antenna Array with Uncorrelated Antennas}
When the separation between two antennas in the array, $\Delta$, exceeds approximately half a wavelength, the channels can be well-modeled as uncorrelated. In this case, the array behaves as an uncorrelated fixed-antenna system moving through spatially correlated positions. Note that the basic approach in \eqref{s*} still applies here. The LCR expression for such a process, for 
$M$ independent antennas, is derived in \cite{Yacoub_ITVT_2001} and is given in \eqref{lcr_uncorrelated_array}. Using this result, an asymptotically exact lower bound on the cdf of $S^*$ is obtained in Result~\ref{res_array_equal}.
\begin{res}\label{res_array_equal}
An asymptotically exact lower bound on the cdf of the SNR for a MA array employing MRC under uncorrelated antennas is given by
\begin{equation}\label{cdf_uncorrelated_array}
    \text{F}_{S^*}(s_{\text{th}}) \geq \frac{\gamma\left(M,\frac{\sigma^2s_\text{th}}{E_{x} \beta}\right)}{\Gamma(M)}- T\times LCR^{(\text{unc})}_{S(t)}(s_\text{th}), 
\end{equation}
where $\gamma(u,n)$ is the lower incomplete gamma function~\cite[Eq. (8.350-1)]{Gradshteyn_book_2007}, $\Gamma(.)$ is the gamma function~\cite[Eq. (8.350-3)]{Gradshteyn_book_2007}, and $LCR^{(\text{unc})}_{S(t)}(s_\text{th})$ is given by 
\begin{align}\label{lcr_uncorrelated_array}
LCR^{(\text{unc})}_{S(t)}(s_\text{th}) = \frac{\sqrt{b}}{\sqrt{2 \pi} (M-1)!}\left(\frac{\sigma ^2s_\text{th}}{E_{x}\beta}\right)^{2M-1}e^{\frac{-s_\text{th} \sigma^2}{E_{x}\beta}}.
    \end{align}
\end{res}   
\begin{proof}
 $LCR^{(\text{unc})}_{S(t)}(s_\text{th})$ is derived in \cite{Yacoub_ITVT_2001}. Since each $|h_i(t)|^2$ is exponentially distributed with mean $\beta$, their sum $\sum_{i=1}^M |h_i(t)|^2$ follows a gamma distribution i.e., $\sum_{i=1}^M |h_i(t)|^2\sim \Gamma(M,\beta)$. Substituting the above results into \eqref{s*} yields \eqref{cdf_uncorrelated_array}.
\end{proof}
\subsection{Movable Antenna Array with Correlated Antennas}
In practice, the available space to implement the MA might be limited, requiring the antennas in the array to be placed in close proximity. Furthermore, it is shown in Section~\ref{4} that a small antenna spacing improves system performance. Therefore, it is important to consider the system in the presence of spatial correlation between antennas. To analyze the performance of such a system, the same approach as \eqref{s*} can be used. It is important to point-out that the analytical expression for the LCR of such an SNR process is novel and has not appeared before in the literature. In this paper, we derive a new closed-form expression for the LCR of SNR under this correlated antenna model.
\begin{theorem}\label{LCR_new}
The LCR of the SNR for an MRC system with $M$ correlated antennas moving through spatially correlated positions is given by
\begin{align}\label{corr_array_lcr}
LCR^{\text{(corr)}}_{S(t)}(s_\text{th}) = \frac{1}{4\sqrt{2}\pi|\mathbf{Q}|} \int_{-\infty}^{\infty} \sum_{i=1}^{M}A_ig_i^{-\frac{3}{2}}e^{-jt_1 s_{\text{th}}} dt_1,
    \end{align} 
    where $\mathbf{Q} = 4 \big(\frac{E_x\beta}{\sigma^2}\big)^2\mathbf{\Sigma}^{\frac{1}{2}}\mathbf{B}\mathbf{\Sigma}^{\frac{1}{2}}$, $g_i$ are the eigenvalues of $\mathbf{G} = \mathbf{Q}^{-1} - jt_1\mathbf{Q}^{-\frac{1}{2}} \mathbf{\Sigma} \mathbf{Q}^{-\frac{1}{2}}$, $A_i = \prod_{j=1,j\neq i}^{M} (g_j - g_i)^{-1}$ and the matrix $\mathbf{B}$ is given by 
\begin{equation}
[\mathbf{B}]_{r,s} = 
\begin{cases}
\pi^2, & \text{if } r = s \\
\displaystyle \frac{\pi J_1\left(2\pi |r-s|\Delta\right)}{|r-s|\Delta}, & \text{if } r \neq s.
\end{cases}
\end{equation}
\end{theorem}   
\begin{proof}
    See Appendix.
\end{proof}
\begin{cor}
   The cdf of $S^*$ is lower bounded by
\begin{align}\label{cdf_corr}
&F_{S^*}(s_{\text{th}}) \geq 1 - \sum_{i=1}^{M} \frac{e^{ \frac{-s_\text{th}\sigma^2}{\lambda_i\beta E_{x}}}}{P_i} - T\times LCR^{\text{(corr)}}_{S(t)}(s_{\text{th}}),
\end{align}  where $\lambda_i$ are the eigenvalues of ~$\mathbf{\Sigma}$ and $P_i = \prod_{\substack{j=1 \\ j \neq i}}^{M} (1-\lambda_j/\lambda_i)$.

\end{cor}
\begin{proof}
    $LCR^{\text{(corr)}}_{S(t)}(s_{\text{th}})$ is given in \eqref{corr_array_lcr}. Since the antennas in the array are correlated, the $|h_i(t)|^2$ are no longer independent. In this case, we can write $\mathbf{h}(t) = \sqrt{\beta}\, \mathbf{\Sigma}^{1/2} \mathbf{u}(t),$ where $\mathbf{u}(t) \sim \mathcal{CN}(\mathbf{0}, \mathbf{I}_M)$. Using the eigen-decomposition for $\mathbf{\Sigma}$, we write
\begin{equation}
    \mathbf{\Sigma} = \mathbf{U} \mathbf{\Lambda} \mathbf{U}^H,
\end{equation}
where $\mathbf{U}$ is a unitary matrix whose columns are the eigenvectors of $\mathbf{\Sigma}$, and $\mathbf{\Lambda} = \mathrm{diag}(\lambda_1, \dots, \lambda_M)$ contains the corresponding eigenvalues. Now $S(t)$ can be written as
\begin{equation}
    S(t) = \frac{E_x}{\sigma^2} \mathbf{h}^H(t) \mathbf{h}(t) 
         = \frac{\beta E_x}{\sigma^2} \sum_{i=1}^{M} \lambda_i \, |\tilde{u}_i|^2,
\end{equation}
where $\tilde{\mathbf{u}} = \mathbf{U}^H \mathbf{u}(t)$ has independent standard complex Gaussian entries, $(\tilde{\mathbf{u}})_i$ = $\tilde{u}_i$. Hence, $S(t)$ becomes a hypoexponential random variable and the cdf is given in \cite{smaili_2013_hypoexponential}. 
\end{proof}
\section{Numerical results}\label{4}
In this section, we validate our theoretical analysis of MA arrays with Monte Carlo simulations. For the sake of presentation, we set $\sigma^2 =1$, $b=\pi^2$, and $E_x\beta=1$.

Fig.~\ref{LCR} shows the LCR of SNR versus the threshold $s_{\text{th}}$ for three configurations: $M=4, \Delta=0.25$; $M=6, \Delta=0.25$; and $M=4, \Delta=0.5$, with movable length is set at $T=1$. The analytical LCR curves, obtained from \eqref{corr_array_lcr}, match the simulation results exactly for all scenarios. For a fixed number of elements ($M=4$), a smaller separation ($\Delta=0.25$) leads to a higher spatial correlation, causing rapid SNR fluctuations. At a fixed separation ($\Delta=0.25$), a smaller array ($M=4$) has a higher LCR at low SNR thresholds, while a larger array ($M=6$) has a higher LCR at high thresholds. This occurs because smaller arrays have lower average SNR, leading to more frequent low-threshold crossings, whereas larger arrays more often cross higher thresholds.
\begin{figure}[!t]\centering
  \includegraphics[width=\linewidth]{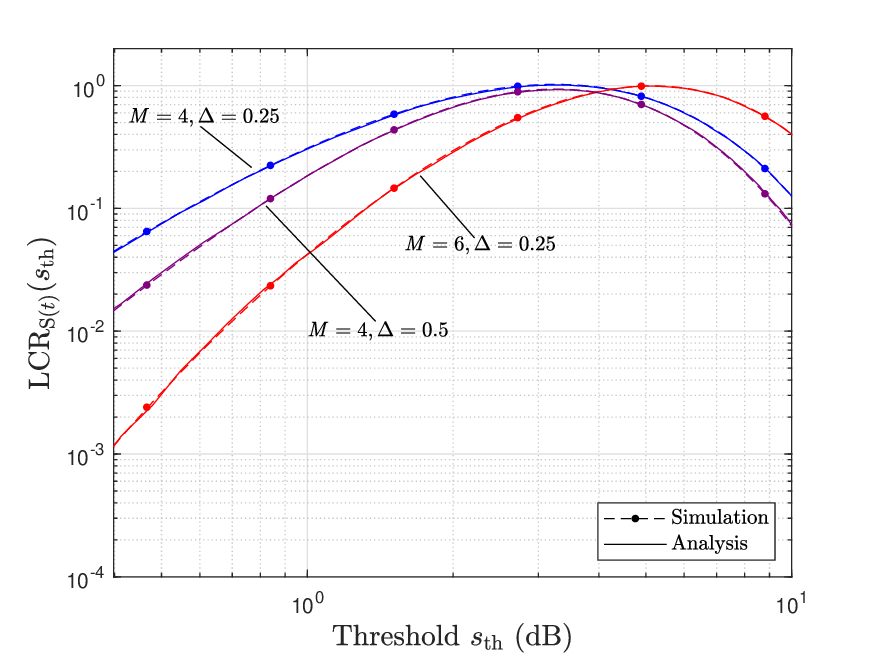}
    \caption{The LCR of the SNR of moving arrays with different dimensions.}
    \captionsetup{justification=centering}
    \label{LCR}
\end{figure}

For the same three configurations, Fig.~\ref{ccdf} plots the complementary cdf (ccdf) of the received SNR of the MA array. The analytical results are generated using the lower bound in \eqref{cdf_corr}. From the plot, we observe that the analytical result accurately lower bounds the simulation curves in the high threshold SNR regime. This regime corresponds to the probability of achieving high instantaneous SNRs, which determines peak data rates and system throughput. Increasing the number of fixed antennas in the array improves the system performance as expected. Furthermore, we observe that system performance improves as the antenna spacing (i.e., $\Delta$) decreases. This is because, at the optimal array position, all antennas encounter favorable channel conditions simultaneously due to high spatial correlation. In contrast, with larger spacing, not all antennas may have strong channels at the same time, resulting in lower overall performance. Unlike fixed arrays, where spacing reduces correlation to avoid deep fades, movable arrays benefit from compactness, enabling all elements to align with strong channels via mobility. Note that in practical systems, discrete sampling may result in performance slightly lower than the theoretical bound.
\begin{figure}[!t]\centering
  \includegraphics[width=\linewidth]{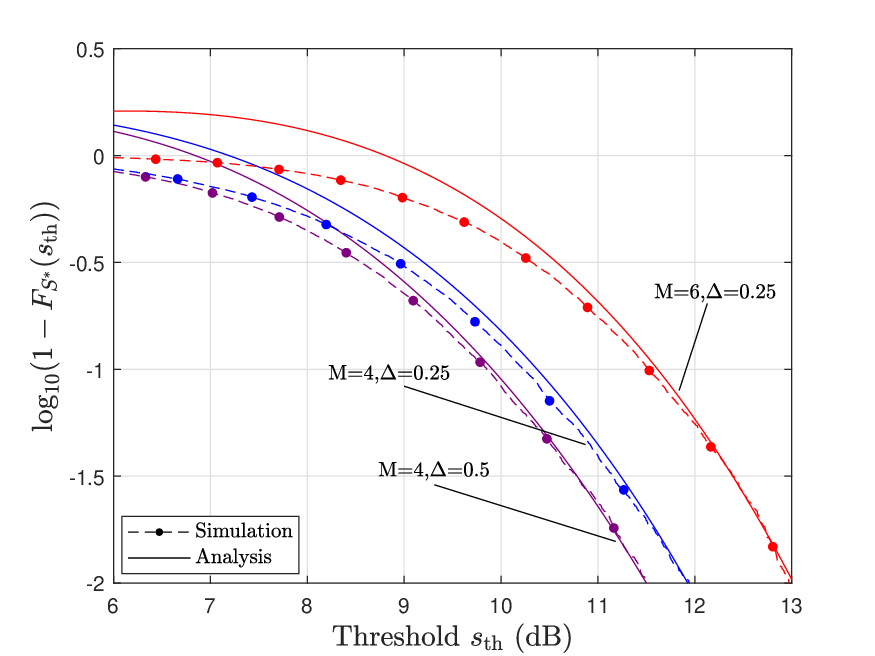}
    \caption{Logarithmic plot of the ccdf of the SNR of moving arrays with different dimensions, with $T=1$.}
    \captionsetup{justification=centering}
    \label{ccdf}
\end{figure}

The proposed MA array has two degrees of freedom: antenna separation ($\Delta$) and movable length ($T$), which can be exploited in different ways. To study their impact, Fig.~\ref{cdftwolen} plots the cdf of $S^*$ for four configurations: $T=1, \Delta=0.1$; $T=1, \Delta=0.5$; $T=0.1, \Delta=0.1$; and $T=0.1, \Delta=0.5$. As shown, increasing $T$ always improves performance. Furthermore, we can observe that for a fixed $\Delta$, increasing $T$ has a greater effect on outage than the high SNR region. It can be observed that increasing $\Delta$ reduces the outage due to increased spatial diversity. For a nearly static array (e.g., $T=0.1$), $\Delta$ does not affect the mean SNR. Rather, increasing $\Delta$ stabilizes the SNR, which reduces variability and lowers peak values. This explains the behavior observed in the plot for $T=0.1$, where spatial diversity is advantageous for lower threshold values while the correlation gain is advantageous in the high SNR region. Hence, for nearly static arrays, there is a trade-off between spatial diversity and correlation gain of antenna elements. When $T$ is large enough, mobility becomes the dominant factor in reducing outage, having a greater effect than spatial separation. With a large $T$, a single antenna in the array can reach positions with strong channel power, and a smaller $\Delta$ ensures that all antennas occupy this high-power region, explaining the observed behavior for the $T=1$ scenarios. Therefore, unneeded spatial diversity can be exchanged for high power correlated locations. Note that this property is not trivial as even for small $T$ (i.e., $T$=1), reducing $\Delta$ from 0.5 to 0.1 makes a substantial improvement. Moreover, this is advantageous as a more compact array reduces the overall dimensions of the system.

\begin{figure}[!t]\centering
  \includegraphics[width=\linewidth]{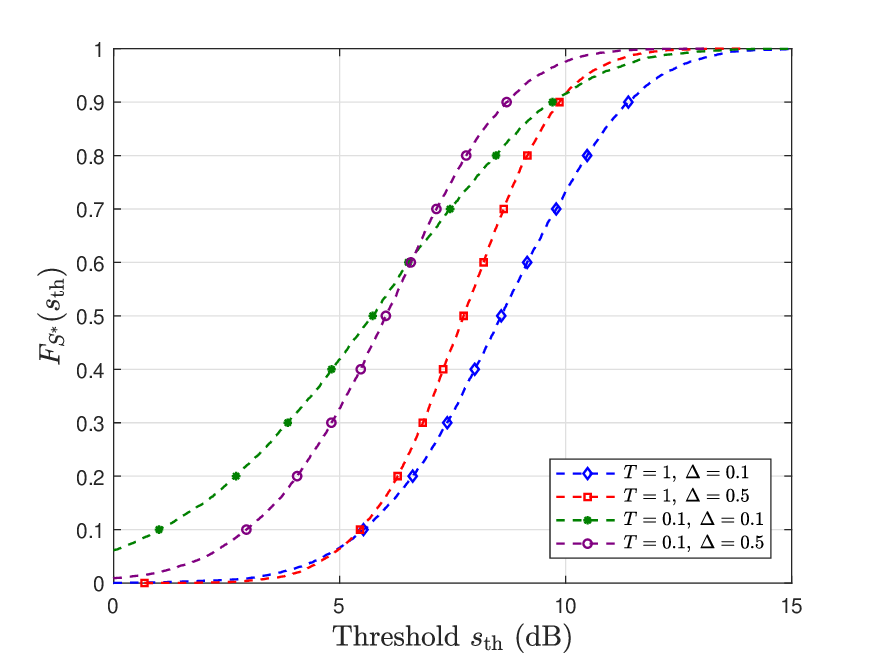}
    \caption{The cdf of the SNR for moving arrays with different dimensions at movable lengths $T=1
    $ and $T=0.1$.}
    \captionsetup{justification=centering}
    \label{cdftwolen}
\end{figure}

Finally, Fig.~\ref{comparison} compares the cdf of the received SNR for five systems: a single fixed antenna (SA), a single fluid antenna (SFA), a four-element fixed array (FA\_4) with $\Delta=0.5$, and two four-element MA arrays (MAA\_4) with $\Delta=0.5$ and $\Delta=0.1$. MRC is used for the three multi-antenna systems and all movable antennas move within one-wavelength region (i.e., $T=1$). The performance of the systems improves as more sources of spatial variability are exploited. The SA has the lowest performance since it lacks both mobility and diversity. The SFA provides considerable improvement due to mobility. In contrast, the FA\_4 achieves higher performance through MRC of the received signals, even without mobility. 
Both four-element MA arrays (MAA\_4), with $\Delta=0.5$ and $\Delta=0.1$, further improve performance by combining MRC with linear mobility. In particular, the configuration with smaller spacing ($\Delta=0.1$) achieves better performance, highlighting the additional benefit of reduced inter-element separation.
\vspace{-5pt}  
\begin{figure}[tb]\centering
  \includegraphics[width=\linewidth]{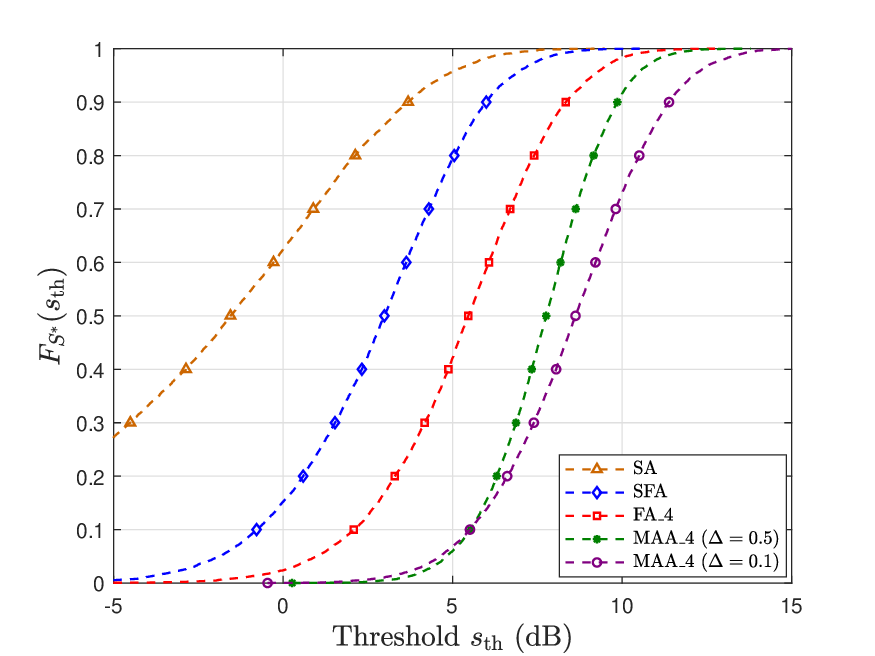}
    \caption{Performance comparison of (a) SA, (b) SFA, (c) FA\_4 with $\Delta=0.5$, (d) MAA\_4 with $\Delta=0.5$, and (e) MAA\_4 with $\Delta=0.1$ configurations. 
  All movable antennas move within a one-wavelength region ($T=1$).}
    \captionsetup{justification=centering}
    \label{comparison}
\end{figure}

\section{conclusion}
In this paper, we consider a MA array where the elements of a linear array are fixed relative to each other, while the entire array moves along an axis perpendicular to its orientation. We derive a novel closed-form expression for the LCR of the SNR when the array elements are correlated. Additionally, we obtain a lower bound on the cdf of the SNR’s supremum for this system. Our analysis is validated through Monte Carlo simulations, which confirms both the accuracy of the LCR expression and the effectiveness of the cdf bound in the upper tail. The results further demonstrate that reducing the spacing between elements leads to improved system performance.

\appendix \label{LCR_corr_de}
First, we denote the covariance matrix between the channel vectors $\mathbf{h}(t)$ and $\mathbf{h}(t+\tau)$ as 
$\mathbf{\Sigma}_{\tau}$. The $(r,s)$-th entry of the matrix $\mathbf{\Sigma}_{\tau}$ corresponds to the correlation between the channel of the $r$-th antenna at position $t$ and the $s$-th antenna at position $t+\tau$. For convenience, we assume $\beta=1$ and generalize to arbitrary powers later. By using \eqref{dist} and \eqref{correlation} we can write the $(r,s)$-th entry of the matrix $\mathbf{\Sigma}_{\tau}$ as
\begin{align}
    [\mathbf{\Sigma}_{\tau}]_{r,s} &= \mathbb{E}[h_r(t) h_s^*(t+\tau)]\nonumber\\
    &= J_0(2\pi \sqrt{\tau^2 + \Delta^2(r-s)^2})\nonumber\\
    &= J_0(2\pi|r-s|\Delta + \delta) + o(\tau^2),
\end{align}where $\delta = \frac{\pi \tau^2}{\Delta|r-s|}.$ Using the first-order expansion of the zeroth-order Bessel function,  
\begin{align}
    \frac{J_0(c+\epsilon) - J_0(c)}{\epsilon} \approx -J_1(c) \quad \text{as} \: \epsilon \to 0,
\end{align}
we obtain, for $r \neq s$,  
\begin{align}\label{sigtaurs}
    [\mathbf{\Sigma}_{\tau}]_{r,s} 
    &= J_0\!\big(2\pi |r-s| \Delta\big) 
    - J_1\!\big(2\pi |r-s| \Delta\big)\,\delta 
    + o(\tau^2).
\end{align}
For $r=s$, by using \eqref{dist} and \eqref{correlation}, we can write
\begin{align}\label{sigtaurr}
[\mathbf{\Sigma}_{\tau}]_{r,r} 
    &= J_0(2\pi \tau) = 1 - \pi^2 \tau^2 + o(\tau^2).
\end{align} 
Now, using \eqref{sigmars}, \eqref{sigtaurs}, and \eqref{sigtaurr} we can write $\mathbf{\Sigma}_{\tau}$ as, 
\begin{align}\label{sigxy}
    \mathbf{\Sigma}_{\tau} = \mathbf{\Sigma} - \mathbf{B}\tau^2,
\end{align}
where, $[\mathbf{B}]_{r,s} = \frac{\pi J_1(2\pi|r-s|\Delta)}{|r-s|\Delta}$ and $[\mathbf{B}]_{r,r} = \pi^2$. For a correlated Gaussian channel, the spatial variation of the channel vector can be modeled as
\begin{equation}\label{ht+tau}
    \mathbf{h}(t+\tau) = \mathbf{A}\mathbf{h}(t)+\mathbf{e},
\end{equation}
where $\mathbf{A}$ and $\mathbf{e}$ control the correlation, while $\mathbf{e} \sim \mathcal{CN}(\mathbf{0}, \mathbf{I}_M)$ is a vector of independent unit-variance complex Gaussian entries. Since the covariance matrix at $t+\tau$ is still equal to $\mathbf{\Sigma}$, from \eqref{ht+tau} we can write
\begin{align}\label{corrsig}
    \mathbf{\Sigma} = \mathbf{A} \mathbf{\Sigma} \mathbf{A}^H +  \mathbb{E}(\mathbf{e}\mathbf{e}^H).
\end{align}
Furthermore, $\mathbf{\Sigma}_{\tau}$ can be written as 
\begin{align}\label{sigtau}
    \mathbf{\Sigma}_{\tau}=\mathbb{E}[\mathbf{h}(t)\mathbf{h}(t+\tau)^H].
\end{align}
By substituting the expression for $\mathbf{h}(t+\tau)$ from \eqref{ht+tau}, \eqref{sigtau} can be expressed as
\begin{align}\label{Amat}
\mathbf{\Sigma}_{\tau} &= \mathbf{\Sigma} \mathbf{A}^H.
\end{align}
Now, using \eqref{snr_array}, $\dot{S}(t)$ can be written as
\begin{align}
    \dot{S}(t) &= \lim_{\tau \to 0}\frac{E_{x}\beta \big(\mathbf{h}(t+\tau)^H \mathbf{h}(t+\tau) -\mathbf{h}(t)^H \mathbf{h}(t)\big) }{\sigma^2 \tau},
\end{align}
for any $\beta$. Then, by substituting from \eqref{ht+tau}, this becomes
\begin{align}\label{dots}
    \dot{S}(t) &= \lim_{\tau \to 0}\frac{E_{x}\beta} {\sigma^2 \tau} \big(\mathbf{h}(t)^H\mathbf{A}^H\mathbf{A}\mathbf{h}(t) + \mathbf{e}^H \mathbf{e}+ \mathbf{e}^H \mathbf{A}\mathbf{h}(t) \nonumber\\
    &+\mathbf{h}(t)^H\mathbf{A}^H\mathbf{e} - \mathbf{h}(t)^H\mathbf{h}(t) \big).
\end{align}
To obtain an expression for $\dot{S}(t)$ using \eqref{dots}, we first derive an expression for $\mathbf{A}^H\mathbf{A}$ using  \eqref{sigxy} and \eqref{Amat} as
\begin{align}\label{AA}
    \mathbf{A}^H\mathbf{A} &= \mathbf{\Sigma}^{-1}\mathbf{\Sigma}_{\tau}\mathbf{\Sigma}_{\tau}^H\mathbf{\Sigma}^{-1}\nonumber\\
    &= \mathbf{I} + \mathbf{C} \tau^2 + o(\tau^4) ,
\end{align}
where $\mathbf{C} = -\mathbf{B}\mathbf{\Sigma}^{-1} -  \mathbf{\Sigma}^{-1}\mathbf{B}$.
Next, using  \eqref{sigxy}, \eqref{corrsig}, and \eqref{Amat}, we write
\begin{align}
    \mathbb{E}[\mathbf{e}\mathbf{e}^H] &= \mathbf{\Sigma} - (\mathbf{\Sigma} - \mathbf{B} \tau^2)\mathbf{\Sigma}^{-1} (\mathbf{\Sigma} - \mathbf{B} \tau^2)\nonumber\\
    &=  2 \mathbf{B} \tau^2 + o(\tau^4).
\end{align}
Hence, $\mathbf{e} = \sqrt{2} \tau \mathbf{B}^{\frac{1}{2}}\mathbf{u}$ where $\mathbf{u} \sim \mathcal{CN}(\mathbf{0}, \mathbf{I}_M)$.
Using this expression for $\mathbf{e}$ and the expression in \eqref{AA}, $\dot{S}(t)$ becomes
\begin{align}
\dot{S}(t)
&= \frac{E_x\beta}{\sigma^2}\lim_{\tau \to 0}\frac{\mathbf{e}^H \mathbf{A} \mathbf{h}(t) + \mathbf{h}(t)^H \mathbf{A}^H \mathbf{e}}{\tau}\nonumber\\
&=  \frac{E_x\beta}{\sigma^2}\lim_{\tau \to 0}\frac{\sqrt{2}\tau(\mathbf{u}^H \mathbf{B}^{\frac{1}{2}}\mathbf{A}\mathbf{h}(t) + \mathbf{h}(t)^H \mathbf{A}^H \mathbf{B}^{\frac{1}{2}}\mathbf{u})}{\tau}.
\end{align}
From \eqref{AA}, it follows that  $\mathbf{A} \to \mathbf{I}$ as $\tau \to 0$. Therefore,
\begin{align}
 \dot{S}(t) &= \frac{\sqrt{2} E_x\beta}{\sigma^2}(\mathbf{u}^H \mathbf{B}^{\frac{1}{2}}\mathbf{h}(t) + \mathbf{h}(t)^H\mathbf{B}^{\frac{1}{2}}\mathbf{u})\nonumber\\
 &= \frac{2\sqrt{2} E_x\beta}{\sigma^2} \Re\!\left\{ \mathbf{u}^H \mathbf{B}^{\tfrac{1}{2}} \mathbf{h}(t) \right\},
\end{align}
where $\Re\{x\}$ represents the real part of $x$. Now we can write $\dot{S}(t) \sim \mathcal{N}(0, \sigma_0^2)$, where $\sigma_0^2$ can be calculated as
\begin{align}
    \sigma_0^2 &= 8 \bigg(\frac{E_x\beta}{\sigma^2}\bigg)^2 \mathbb{E}\left[\frac{|\mathbf{u}^H \mathbf{B}^{\frac{1}{2}}\mathbf{h}(t)|^2}{2}\right]\nonumber\\
    &= 4 \bigg(\frac{E_x\beta}{\sigma^2}\bigg)^2  \operatorname{tr}\bigg(\mathbf{B}^{\frac{1}{2}}\mathbf{h}(t)\mathbf{h}(t)^H\mathbf{B}^{\frac{1}{2}}\bigg)\nonumber\\
    &= 4 \bigg(\frac{E_x\beta}{\sigma^2}\bigg)^2\mathbf{h}(t)^H \mathbf{B} \mathbf{h}(t),
\end{align}
where $\text{tr}(.)$ represents the trace of the given matrix. Since $\mathbf{h}(t) = \mathbf{\Sigma} ^ {\frac{1}{2}} \mathbf{v}$ with $\mathbf{v} \sim \mathcal{CN}(\mathbf{0}, \mathbf{I}_M)$, we can write $\dot{S}(t) = \sqrt{\mathbf{{v}}^H\mathbf{Q}\mathbf{{v}}} z$ where $z \sim \mathcal{N}(0,1)$ and $\mathbf{Q} = 4\bigg(\frac{E_x\beta}{\sigma^2}\bigg)^2\mathbf{\Sigma}^{\frac{1}{2}}\mathbf{B} \mathbf{\Sigma}^{\frac{1}{2}}$. Defining $p=\mathbf{{v}}^H\mathbf{Q}\mathbf{{v}}$, we obtain $\dot{S} \sim \mathcal{N}(0, p)$, and furthermore, $S(t) = \mathbf{h}(t)^H\mathbf{h}(t) = \mathbf{{v}}^H\mathbf{\Sigma}\mathbf{{v}}$.

The LCR of SNR can be written as
\begin{align}\label{lcrbasic}
    LCR_{S(t)}(s_\text{th}) &=  \int_0^\infty \dot{s}f_{S,\dot{S}}(s_\text{th},\dot{s}) d\dot{s}\nonumber\\
    &= \int_0^\infty \dot{s}\int_0^{\infty} f_{\dot{S}|S,P}(\dot{s}|s_{\text{th}},p) f_{S,P}(s_{\text{th}},p)dp d\dot{s}.
\end{align}
Since $\dot{S} \sim \mathcal{N}(0, p)$, the conditional pdf $f_{\dot{S}|S,P}(\dot{s}|s_{\text{th}},p)=\frac{e^{\frac{-\dot{s}^2}{2p}}}{\sqrt{2\pi p}}$.
Therefore, \eqref{lcrbasic} can be written as 
\begin{align}
     LCR_{S(t)}(s_\text{th}) &= \int_0^{\infty} \int_0^\infty \dot{s}\frac{e^{\frac{-\dot{s}^2}{2p}}}{\sqrt{2\pi p}}d\dot{s} f_{S,P}(s_{\text{th}},p)dp.
\end{align}
Now, using the integral identity in [Eq. 3.326-2]\cite{Gradshteyn_book_2007} we obtain
\begin{align}\label{lcrsp}
     LCR_{S(t)}(s_\text{th}) &= \frac{1}{\sqrt{2\pi}}\int_0^{\infty} \sqrt{p} f_{S,P}(s_{\text{th}},p)dp.
\end{align}
Now, we proceed with the following approach to find the joint probability density function (pdf) of $S$ and $P$, $f_{S,P}(s_{\text{th}},p)$. Since $S=  \mathbf{{v}}^H\mathbf{\Sigma}\mathbf{{v}}$ and $P=\mathbf{{v}}^H\mathbf{Q}\mathbf{{v}}$, $f_{S,P}(s_{\text{th}},p)$ represents a joint pdf of two quadratic forms. By defining $\mathbf{w} = \mathbf{Q}^{\frac{1}{2}}\mathbf{v}$, we can write $S=\mathbf{w}^H \mathbf{F} \mathbf{w}$ and $P=\mathbf{w}^H\mathbf{w}$ where $\mathbf{F} = \mathbf{Q}^{-\frac{1}{2}} \mathbf{\Sigma} \mathbf{Q}^{-\frac{1} {2}}$. The pdf of $\mathbf{w}$ can be written as
\begin{align}\label{pdfw}
    f_{W}(w) = \frac{1}{\pi^M |\mathbf{Q}|} e^{-(\mathbf{w}^H \mathbf{Q}^{-1}\mathbf{w})}.
\end{align}
Now, using \eqref{pdfw}, the joint characteristic function (jcf) of $S$ and $P$ can be written as
\begin{align}
    \Phi_{S,P}(t_1,t_2)&=\mathbb{E}\left[e^{jt_1S + jt_2P}\right]  \nonumber\\
    &= \mathbb{E}\left[e^{jt_1\mathbf{w}^H \mathbf{F}\mathbf{w} + jt_2\mathbf{w}^H\mathbf{w}}\right]\nonumber\\
    &= \int_{-\infty}^{\infty} e^{\mathbf{w}^H(jt_1\mathbf{F}+jt_2\mathbf{I})\mathbf{w}} f_W(w) dw.
\end{align}
By using the fact that $\int_{-\infty}^{\infty} \frac{e^{\mathbf{-w}^H(\mathbf{Q}^{-1}-jt_1\mathbf{F}-jt_2\mathbf{I})\mathbf{w}}}{\pi^M |(\mathbf{Q}^{-1}-jt_1\mathbf{F}-jt_2\mathbf{I})^{-1}|} dw =1$, $ \Phi_{S,P}(t_1,t_2)$ can be written as
\begin{align}
   \Phi_{S,P}(t_1,t_2) &= \frac{1}{|\mathbf{Q}| |(\mathbf{Q}^{-1}-jt_1\mathbf{F}-jt_2\mathbf{I})|}\nonumber\\
   &= \frac{1}{|\mathbf{Q}| \prod_{i=1}^{M}(g_i - jt_2) },
\end{align}
where $\mathbf{G}=\mathbf{Q}^{-1} - jt_1 \mathbf{F}$ and the eigenvalues of $\mathbf{G}$ are denoted by $g_i$. Note that $g_i$ are functions of $t_1$. Since the joint pdf is the inverse of the jcf, we can write
\begin{align}\label{jcf}
f_{S,P}(s_{\text{th}},p) 
&= \frac{1}{(2\pi)^2} \iint_{\mathbb{R}^2} 
\Phi_{S,P}(t_1,t_2) \, e^{-j (t_1 s_{\text{th}} + t_2 p)} \, dt_1 dt_2 \nonumber\\
&= \frac{1}{(2\pi)^2} \iint_{\mathbb{R}^2} 
\frac{e^{-j (t_1 s_{\text{th}} + t_2 p)}}{|\mathbf{Q}| \prod_{i=1}^{M} (g_i - j t_2)} \, dt_1 dt_2.
\end{align}
By substituting \eqref{jcf} into \eqref{lcrsp}, the LCR of the SNR becomes
\begin{align}
LCR_{S(t)}(s_\text{th}) 
= K
\int_{-\infty}^{\infty} \int_0^{\infty} \int_{-\infty}^{\infty} 
\frac{\sqrt{p} \, e^{-j(t_1 s_\text{th} + t_2 p)}}{\prod_{i=1}^{M} (g_i - j t_2)} 
\, dt_2 \, dp \, dt_1,
\end{align}
where $K=\frac{(2\pi)^{-\frac{5}{2}}}{|\mathbf{Q}|}$. Using partial fraction decomposition, we write $\frac{1}{\prod_{i=1}^{M}(g_i - jt_2)} = \sum_{i=1}^{M} \frac{A_i}{g_i -jt_2}$ where $A_i = \frac{1}{\prod_{k=1,k\neq i}^{M}(g_k - g_i)}$. Then, using the integral identities in [Eq. (3.382-7) and Eq. (3.326-2),~\!\citenum{Gradshteyn_book_2007}], we derive \eqref{corr_array_lcr}.
\bibliographystyle{IEEEtran}
\bibliography{references}

\end{document}